\documentclass[11pt]{article}
\usepackage{amssymb,amsmath,graphicx,amsthm,mathrsfs}
\usepackage[usenames]{color}
\newtheorem{Theorem}{Theorem}[section]
\newtheorem{Lemma}[Theorem]{Lemma}
\newtheorem{Proposition}[Theorem]{Proposition}
\newtheorem{Definition}[Theorem]{Definition}

\newtheorem{Remark}[Theorem]{Remark}
\numberwithin{equation}{section}

\title{\bf Virtual Reshaping and Invisibility in Obstacle Scattering}

\author{Hongyu Liu\thanks{Department of Mathematics,
University of Washington, Box 354350, Seattle, WA 98195, USA. ({\tt
hyliu@math.washington.edu})}}

\begin{document}

\date{}

\maketitle

\begin{abstract}
We consider reshaping an obstacle virtually by using transformation
optics in acoustic and electromagnetic scattering. Among the general
virtual reshaping results, the virtual minification and virtual
magnification are particularly studied. Stability estimates are
derived for scattering amplitude in terms of the diameter of a small
obstacle, which implies that the limiting case for minification
corresponds to a perfect cloaking, i.e., the obstacle is invisible
to detection.
\end{abstract}

\section{Introduction}

Since the pioneering work on transformation optics and cloaking
\cite{GLU,GLU2,Leo,PenSchSmi}, there is an avalanche of study on
designs of various striking cloaking devices; e.g., invisibility
cloaking devices \cite{GKLU3,KSVW}; field rotators \cite{CheCha};
concentrators \cite{LCZRK}; electromagnetic wormholes
\cite{GKLU1,GKLU2}; superscatterers \cite{YCLM}, etc.. We refer to a
most recent survey paper \cite{GKLU4} for a comprehensive review and
related literature. The crucial observation is that certain PDEs
governing the wave phenomena are form-invariant under
transformations, e.g., Hemholtz equation for acoustic scattering and
Maxwell's equations for electromagnetic scattering. Hence, one could
form new acoustic or EM material parameters (in the physical space)
by pushing forward old ones (in the virtual space) via a mapping
$F$. Such materials/media are called \emph{transformation media}
\cite{PenSchSmi}. It turns out that the wave solutions in the
virtual space with the old material parameters and in the physical
space with the new material parameters are also related by the
push-forward $F$. Those key ingredients pave the way for the design
of optical devices with customized effects on wave propagation.

In this paper, we shall be concerned with cloaking devices for
acoustic and electromagnetic obstacle scattering. As is known, there
are two types of scatterers which are under wide study for acoustic
and electromagnetic scattering, namely, the {\it penetrable medium}
and the {\it impenetrable obstacle}. For a medium, the acoustic or
EM wave can penetrate inside, and basically the medium accounts for
the coefficients in the governing PDEs. Whereas for an obstacle, the
acoustic or EM wave cannot penetrate inside and only exists in the
exterior of the object, and the obstacle is related to the domain of
definitions for the governing PDEs. The cloakings for acoustic or EM
media have been extensively studied in transformation optics in
existing literature and the theory has been well-established, we
again refer to the review paper \cite{GKLU4} for related discussion.
For our current study, the cloakings for obstacles are considered
and it is shown that the domain of definitions for certain PDEs can
also be pushed forward under transformations. Using the
transformation optics, one can push forward an obstacle in the
virtual space to form a different obstacle in the physical space,
and the ambient space around the virtual obstacle is then pushed
forward to a cloaking medium around the physical obstacle. With a
suitable push-forward $F$, it is shown that the scattering amplitude
in the physical space coincides with that in the virtual space. That
is, if one intends to recover the physical obstacle after being
cloaked by the corresponding scattering measurements, then the
reconstruction will give the image of the obstacle in the virtual
space, but not the physical one, namely, the physical obstacle is
virtually reshaped with the cloaking. Principally, it has been shown
that one can achieve any desired virtual reshaping effect provided
an appropriate transformation $F$ can be found between the virtual
space and the physical space.

Particularly, we consider virtually magnifying and minifying an
obstacle. By magnification, we mean that the size of the virtual
obstacle is larger than that of the underlying physical one. That
is, under acoustic and EM wave detection, the cloaking makes the
obstacle look bigger than its original size. Whereas by
minification, we actually mean virtually shrinking the obstacle,
that is, the size of the virtual obstacle is smaller than that of
the physical one. In the limiting case of minification, the virtual
obstacle collapses to a single point, and this formally corresponds
to a perfect cloaking, namely, the physical obstacle becomes
invisible to detection. We note that in this case, the push-forward
$F$ blows up a single point in the virtual space to a `hole' (which
actually is the physical obstacle) in the physical space. Hence, the
map $F$ is intrinsically singular, and the obtained transformation
medium is inevitably singular. Correspondingly, the transformed PDEs
in the physical space are no longer uniformly elliptic which also
becomes singular. Therefore, in order to rigorously justify the
perfect cloaking, we need to deal with the singular PDEs. Basically,
one would encounter the same problems in treating perfect cloakings
for acoustic or EM medium and several approaches are proposed to
deal with such singularities. For perfect cloaking of conductivity
equation, which can be considered as optics at zero frequency, the
invisibility is mathematically justified in \cite{GLU2} by using the
removability of point singularities for harmonic functions; whereas
an alternative treating is provided in \cite{KSVW}, where {\it
near-invisibility} is introduced from a regularization viewpoint and
the invisibility is rigorously justified based on certain stability
estimates for conductivity equation with small inclusions. For the
finite frequency cases, a novel notion of \emph{finite energy
solutions} is introduced in \cite{GKLU} and the invisibility
cloaking of acoustic and electromagnetic medium are then justified
directly. For the perfect cloaking of obstacles considered in the
present paper, we shall follow the approach in \cite{KSVW} to
mathematically justify the invisibility by taking limit of
near-invisibility. To that end, we derive certain stability
estimates for scattering amplitudes in terms of the diameter of a
small obstacle in both acoustic and EM scattering. Those stability
estimates are then used to show that the limiting process of
minification cloaking corresponds to a process of near-invisibility
cloaking, which in turn implies the desired invisibility result of
the perfect cloaking. For practical considerations, all our
reshaping studies are conducted within multiple scattering, that is,
there is more than one obstacle component included.

Finally, we would like to mention some unique determination results
in inverse obstacle scattering, where one utilizes acoustic or
electromagnetic scattering measurements to identify an
unknown/inaccessible obstacle. The uniqueness/identifiability
results correspond to circumstances under which one cannot virtually
reshape an obstacle. In the case that the obstacle is situated in a
homogeneous background medium, the uniqueness theory for inverse
obstacle scattering is relatively well established, and we refer to
\cite{LiuZou} for a survey and relevant literature. Whereas in
\cite{KirPai},\cite{Liu},\cite{NPT}, the recovery of an obstacle
included in certain inhomogeneous (isotropic) medium is considered.
It is shown in \cite{KirPai} and \cite{NPT} that if the isotropic
medium is known a priori, then the included obstacle is uniquely
determined by the associated scattering amplitude. Under the
assumption that the isotropic medium and the included obstacle has
only planar contacts, it is proved in \cite{Liu} that one can
recover both the medium and the obstacle by the associated
scattering amplitude. The argument in \cite{Liu} also implies that
an obstacle surrounded by an isotropic medium cannot produce the
same scattering amplitude as another pure obstacle. This result
essentially indicate that transformation media for virtually
reshaping an obstacle must be anisotropic.

The rest of the paper is organized as follows. In
Section~\ref{sect:acoustic reshaping}, we consider the reshaping for
acoustic scattering, where virtual minification and magnification
are first considered consecutively, and then we present a general
reshaping result. Similar study has been conducted for reshaping a
EM perfectly conducting obstacle in Section~\ref{sect:em reshaping}.

\section{Virtual reshaping for acoustic scattering}\label{sect:acoustic
reshaping}

\subsection{The Helmholtz equation}\label{subset:11}

Let $M$ be an open subset of $\mathbb{R}^3$ with Lipschitz
continuous boundary $\partial M$ and connected complement
$M^+:=\mathbb{R}^3\backslash \overline{M}$. Let $(M^+, g)$ be a
Riemannian manifold such that $g$ is Euclidean outside of a
sufficiently large ball $B_R$ containing $M$. Here and in the
following, $B_R$ shall denote an Euclidean ball centered at origin
and of radius $R$. In wave scattering, $M$ denotes an impenetrable
obstacle and the Remannian metric $g$ corresponds to the surrounding
medium with the Euclidean metric $g_0:=\delta_i^j$ representing the
vacuum. In acoustic scattering, $\sigma=(\sigma^{ij})_{i,j=1}^3$
with $\sigma^{ij}:=\sqrt{|g|}g^{ij}$ is the anisotropic acoustic
density and $\sqrt{|g|}=|\sigma|$ is the bulk modulus, where
$(g^{ij})_{i,j=1}^3$ is the matrix inverse of the matrix
$(g_{ij})_{i,j=1}^3$, and $|g|=\det g, |\sigma|=\det \sigma$.
Formally, we have the following one-to-one correspondence between a
material parameter tensor and a Riemannian metric
\begin{equation}\label{eq:relation}
\sigma^{ij}=|g|^{1/2}g^{ij}\quad \mbox{or}\quad
g^{ij}=|\sigma|^{-1}\sigma^{ij}.
\end{equation}

We consider the scattering for a time-harmonic plane incident wave
$u^i=\exp\{i k x\cdot\theta\}$, $\theta\in\mathbb{S}^{2}$ due to the
obstacle $M$ together with the surrounding medium $(M^+, g)$. The
total wave field is governed by the Helmholtz equation
\begin{align}
& \Delta_g u+k^2 u=0\quad \mbox{in $M^+$},\label{eq:Helmholtz}\\
& u|_{\partial M}=0,\label{eq:Dirichlet}
\end{align}
where the Laplace-Beltrami operator associated with $g$ is given in
local coordinates by
\[
\Delta_g u=\frac{1}{\sqrt{g}}\sum_{i,j=1}^3\frac{\partial}{\partial
x_i}\left(\sqrt{|g|}g^{ij}\frac{\partial u}{\partial x_j}\right).
\]
The homogeneous Dirichlet boundary condition (\ref{eq:Dirichlet})
means that the wave pressure vanishes on the boundary of the
obstacle. $M$ is usually referred to as a \emph{sound-soft}
obstacle. The scattered wave field is as usual assumed to satisfy
the \emph{Sommerfeld radiation condition}. Taking advantage of the
one-to-one correspondence (\ref{eq:relation}) between (positive
definite) acoustic densities $\sigma$ and Riemannian metrics $g$, we
proceed to mention a few facts about the form-invariance of the
Hemholtz equation under transformations. For a smooth diffeomorphism
$F:=\Omega_1\rightarrow \Omega_2$, $y=F(x)$, the metric $g(x)$
transforms as a covariant symmetric 2-tensor,
\begin{equation}\label{eq:pushforward}
\tilde{g}_{ij}(y):=(F_* g)_{ij}(y)=\sum_{l,m=1}^{3}\frac{\partial
x^l}{\partial y^i}\frac{\partial x^m}{\partial y^j}
g_{lm}\bigg|_{x=F^{-1}(y)},
\end{equation}
and then, for $u=\tilde{u}\circ F$, we have
\begin{equation}
(\Delta_g+k^2)u=0 \Longleftrightarrow
(\Delta_{\tilde{g}}+k^2)\tilde{u}=0.
\end{equation}
Alternatively, using (\ref{eq:relation}), one could work with the
Helmholtz equation of the following form
\begin{equation}\label{eq:Helmholtz II}
\sum_{i,j=1}^{3}\partial_i(\sigma^{ij}\partial_j u)+k^2 |\sigma|
u=0,
\end{equation}
and then, for $u=\tilde{u}\circ F$, we have
\begin{equation}\label{eq:Helmholtz III}
\sum_{i,j=1}^{3}\partial_i({\tilde\sigma}^{ij}\partial_j u)+k^2
|\tilde\sigma| u=0.
\end{equation}
Here $\tilde{\sigma}$ is the push-forward of $\sigma$ which, by
using (\ref{eq:relation}) and (\ref{eq:pushforward}), can be readily
shown to be given by
\begin{equation}\label{eq:trans sigma}
\tilde{\sigma}=F_* \sigma:=\left(\frac{(D F)^T\cdot \sigma\cdot (D
F)}{|\det D F|}\right)\circ F^{-1},
\end{equation}
where $D F$ denotes the (matrix) differential of $F$ and $(D F)^T$
its transpose.

Throughout, we shall work with $\sigma\in L^\infty(M^+)^{3\times 3}$
and $F$ is orientation-preserving, invertible with both $F$ and
$F^{-1}$ (uniformly) Lipschitz continuous over $M^+$. So, it is
appropriate to work with the following Sobolev space for the
scattering solution to (\ref{eq:Helmholtz})-(\ref{eq:Dirichlet}),
\[
H_{loc}^1(M^+)=\{u\in\mathscr{D}'(M^+); u\in H^1(M^+\cap B_\rho)\ \
\mbox{for each finite $\rho$ with $M\subset B_\rho$}\}.
\]
The system (\ref{eq:Helmholtz})-(\ref{eq:Dirichlet}), or
(\ref{eq:Helmholtz II}) and (\ref{eq:Dirichlet}) is well-posed and
has a unique solution $u\in H_{loc}^1(M^+)$ (see \cite{Mcl}). Noting
that the corresponding metric outside a ball $B_{R}\supset M$ is
Euclidean, we know $u$ is smooth outside $B_{R}$. Furthermore, the
solution $u(x,k,\theta)$ admits asymptotically as $|x|\rightarrow
+\infty$ the development (see \cite{ColKre})
\begin{equation}
u(x,\theta,k)=e^{ik x\cdot \theta}+\frac{e^{ik |x|}}{|x|}
A(\theta',\theta,k)+\mathcal{O}(\frac{1}{|x|^2}),
\end{equation}
where $\theta'=x/|x|\in\mathbb{S}^{2}$. The analytic function
$A(\theta',\theta,k)$ is known as the \emph{scattering amplitude} or
\emph{far-field pattern}. According to the celebrated Rellich's
theorem, there is a one-to-one correspondence between the scattering
amplitude $A(\theta',\theta,k)$ and the wave solution
$u(x,\theta,k)$. Throughout, we consider the scattering amplitude
for the virtual reshaping effects.

We shall denote by $M\bigoplus (M^+,g)$ a cloaking device with an
obstacle $M$ and the corresponding cloaking medium $(M^+,g)$. The
metric $g$ is always assumed to be Euclidean outside a sufficiently
large ball containing $M$, namely, the cloaking medium is compactly
supported. If we know the support of the cloaking medium, say
$M'\backslash \overline{M}$, we also write $M\bigoplus(M'\backslash
\overline{M}, g)$ to denote the cloaking device.

\begin{Definition}\label{def:acoustic}
We say that $M\bigoplus (M^+,g)$ (virtually) reshapes the obstacle
$M$ to another obstacle $\widetilde{M}$, if the scattering
amplitudes coincide for $M\bigoplus (M^+,g)$ and $\widetilde{M}$,
i.e.
\[
A(\theta',\theta,k;M\bigoplus (M^+,g))=A(\theta',\theta,k;
\widetilde{M}).
\]
\end{Definition}

We would like to remark that according to the correspondence
(\ref{eq:relation}), the cloaking device in
Definition~(\ref{def:acoustic}) can also be written as $M\bigoplus
(M^+, \sigma)$, where $\sigma$ is the (anisotropic) acoustic density
for the cloaking medium.

\subsection{Virtual minification by cloaking}

We first consider the reshaping effects for a special class of
obstacles, which are star-shaped and referred to as {\it $l^p$-ball
shaped obstacles} in the following. They are domains in
$\mathbb{R}^3$ of the form
\[
\{x\in\mathbb{R}^3; \|x\|_{p}=r\},
\]
where $p\in [1,+\infty]$, $r>0$ is a constant and for
$x=(x_1,x_2,x_3)$
\[
\|x\|_p=\left(\sum_{i=1}^{3} |x_i|^p\right)^{1/p}.
\]
Obviously, $\|\cdot\|_2=|\cdot|$ and an $l^2$-ball is exactly an
Euclidean ball. For $l^p$-ball shaped obstacles, we can give the
transformation rule explicitly and correspondingly, the cloaking
material parameters for those obstacles can be derived explicitly.
Henceforth, we write $B_{R,p}$ to denote an $l^p$-ball of radius $R$
and centered at origin, whereas as prescribed earlier, we write
$B_{R}:=B_{R,2}$. We also denote by $M^+$ the complement of a domain
$M$ in $\mathbb{R}^3$.

Let $M=B_{R_1,p}$ with $R_1>0$. Let $R_0, R_2$ be such that
$0<R_0<R_1<R_2$. We define the map, $F: B_{R_0,p}^+\mapsto
B_{R_1,p}^+$ by
\begin{equation}\label{eq:1}
x:=F(y)=\begin{cases} \qquad y,\ &\mbox{for $\|y\|_p\geq R_2$},\\
(\frac{R_1-R_0}{R_2-R_0}R_2+\frac{R_2-R_1}{R_2-R_0}\|y\|_p)\frac{y}{\|y\|_p},\
\ & \mbox{for $R_0<\|y\|_p<R_2$}.
\end{cases}
\end{equation}
It is noted that $F$ is (uniformly) Lipschitz continuous over
$B_{R_0,p}^+$ and maps $B_{R_2,p}\backslash B_{R_0,p}$ to
$B_{R_2,p}\backslash B_{R_1,p}$. For $R_1<\|x\|_p<R_2$, let
\begin{equation}\label{eq:1.6}
g_1(x)=(F_*g_0)(x).
\end{equation}
\begin{figure}
\centering
  \includegraphics[width=9cm]{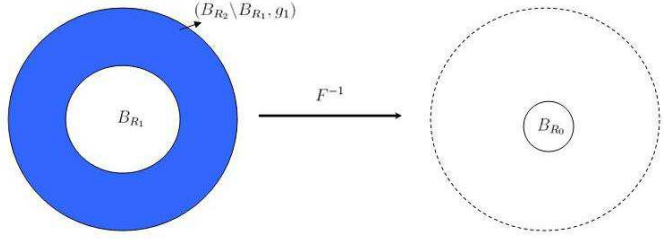}\\
  \caption{Illustration for minification: the cloaking device $B_{R_1}\bigoplus (B_{R_1}^+, g)$
  virtually reshapes $B_{R_1}$ to $B_{R_0}$.}\label{fig:1}
\end{figure}
\noindent Set
\begin{equation}\label{eq:medium1}
g(x)=\begin{cases} g_1(x),\ &\mbox{for $R_1<\|x\|_p<R_2$},\\
g_0(x), & \mbox{for $\|x\|_p\geq R_2$}.
\end{cases}
\end{equation}
We have
\begin{Proposition}\label{prop:1}
The cloaking device $B_{R_1,p}\bigoplus (B_{R_1,p}^+, g)$ with $g$
defined in (\ref{eq:medium1}), reshapes $B_{R_1,p}$ virtually to
$B_{R_0,p}$. That is, the physical obstacle $B_{R_1,p}$ with the
cloaking material $(B_{R_2,p}\backslash B_{R_1,p},g_1)$ is virtually
minified to the obstacle $B_{R_0,p}$ with a minification ratio
$\kappa:=R_0/R_1$ (see Fig.~\ref{fig:1} for a schematic
illustration).
\end{Proposition}

\begin{proof}
Let $u(y)\in H_{loc}^1(B_{R_0,p}^+)$ be the unique solution to the
Hemholtz equation (\ref{eq:Helmholtz})-(\ref{eq:Dirichlet})
corresponding to the obstacle $B_{R_0,p}$. Whereas, we let $v(x)\in
H_{loc}^1(B_{R_1,p}^+)$ be the unique solution to the Helmholtz
equation (\ref{eq:Helmholtz})-(\ref{eq:Dirichlet}) corresponding to
$B_{R_1,p}\bigoplus (B_{R_1,p}^+,$ $ g)$. Define $\tilde{u}(x), x\in
B_{R_1,p}^+,$ be such that $u=\tilde{u}\circ F$, i.e. $\tilde{u}=F_*
u=(F^{-1})^*u$. It is clear that $\tilde u\in
H_{loc}^1(B_{R_1,p}^+)$ since $F$ is bijective and both $F$ and
$F^{-1}$ are (uniformly) Lipschitz continuous. Moreover, noting
$F(\partial B_{R_0,p})=\partial B_{R_1,p}$, we know
$\tilde{u}|_{\partial B_{R_1,p}}=0$.

By the invariance of Helmholtz equation under transformation, it is
readily seen that $\tilde{u}=v$. Hence,
\begin{equation}\label{eq:equal 1}
A(\theta',\theta,k;B_{R_1,p}\bigoplus (B_{R_1,p}^+,
g))=A(\theta',\theta,k;B_{R_0,p}).
\end{equation}
\end{proof}

For an Euclidean ball $B_{R_0}\subset\mathbb{R}^3$, by separation of
variables, we have
\begin{equation}
A(\theta',\theta,k;B_{R_0})=\frac{i}{k}\sum_{n=0}^{\infty}(2n+1)\frac{j_n(k
R_0)}{h_n^{(1)}(k R_0)}P_n(\cos\psi),
\end{equation}
where $j_n(t)$ and $h_n^{(1)}(t)$ are respectively, the $n$-th order
spherical Bessel function and spherical Hankel function of first
kind, $P_n(t)$ is the Legendre polynomial and $\psi=\angle
(\theta,\theta')$.
Using the asymptotical properties
\begin{align*}
j_n(t)=\mathcal{O}(t^n),\ \ h_n^{(1)}(t)=\mathcal{O}(t^{-n-1}),\ \
n=0,1,\ldots,\ \ \mbox{as\ $t\rightarrow +0$},
\end{align*}
it is straightforward to show
\begin{equation}\label{eq:asym 1}
A(\theta',\theta,k;B_{R_0})=\mathcal{O}(R_0)\quad \mbox{as
$R_0\rightarrow +0$,}
\end{equation}
Now we consider the limiting case for minification, namely
$\kappa\rightarrow +0$ or equivalently $R_0\rightarrow +0$.
By (\ref{eq:equal 1}) and (\ref{eq:asym 1}),
\begin{equation}\label{eq:stability 1}
A(\theta',\theta,k;B_{R_1}\bigoplus (B_{R_1}^+,
{g}))=A(\theta',\theta,k; B_{R_0})=\mathcal{O}(R_0)
\end{equation}
as $R_0\rightarrow +0$. That is,
\begin{Proposition}\label{prop:invisible}
The limit for minification of an Euclidean ball $B_{R_1}$ in
Proposition~\ref{prop:1} gives a perfect cloaking, namely, it makes
the obstacle invisible to detection.
\end{Proposition}


In the limiting case with $\kappa=0$, the transformation in
(\ref{eq:1}) becomes
\begin{equation}\label{eq:ideal trans mini}
x=H(y):=\begin{cases} \quad y,\quad \mbox{for $\|y\|_p\geq R_2$},\\
(R_1+\frac{R_2-R_1}{R_2}\|y\|)\frac{y}{\|y\|_p}, &\mbox{for
$0<\|y\|_p<R_2$},
\end{cases}
\end{equation}
which maps $\mathbb{R}^3\backslash\{0\}$ to $\mathbb{R}^3\backslash
B_{R_1,p}$, i.e., it blows up the single point $\{0\}$ to
$B_{R_1,p}$. It is remarked that the map $H$ in (\ref{eq:ideal trans
mini}) with $p=2$ is exactly the one used in \cite{GLU,GLU2} for
perfect cloaking of conductivity equation, and in \cite{PenSchSmi}
for perfect cloaking of electromagnetic material tensors. Next, we
take the case with $p=2$ as an example for a simple analysis of the
perfect cloaking medium. The corresponding metric $(H_*g_0)(x)$ in
$B_{R_2}\backslash\bar{B}_{R_1}$ is singular near the cloaking
interface, namely $\partial B_{R_1}$. In fact, considering in the
standard spherical coordinates on $B_{R_2}\backslash \{0\}$,
$(r,\phi,\theta)\mapsto (r\sin\theta\cos\phi,
r\sin\theta\sin\phi,r\cos\theta)\in\mathbb{R}^3$ and by
(\ref{eq:pushforward}), it can be easily calculated that
\[
\tilde g:=H_*g_0=\left(
      \begin{array}{ccc}
        \lambda^2 & 0 & 0 \\
        0 & \lambda^2(r-R_1)^2 & 0 \\
        0 & 0 & \lambda^2(r-R_1)^2\sin^2\theta \\
      \end{array}
    \right),
\]
where $\lambda=R_2/(R_2-R_1)$. That is, $\tilde{g}$ has one
eigenvalue bounded from below (with eigenvector corresponding to the
radial direction) and two eigenvalues of order $(r-R_1)^2$
approaching zero as $r\rightarrow +R_1$. Hence, if the perfect
cloaking is analyzed directly, one needs to deal with the
degenerated elliptic equation near the cloaking interface. So, a
suitable choice of the class of \emph{weak} solutions to the
singular equation must be purposely introduced, as the \emph{finite
energy solutions} considered in \cite{GKLU} for invisibility
cloaking devices of acoustic and electromagnetic media. Clearly, our
earlier analysis on the perfect cloaking of an Euclidean ball avoid
singular equation by taking limit. This is similar to \cite{KSVW}
for the analysis of perfect cloaking of conductivities in electrical
impedance tomography by regularization. Here we would like to point
out that there is no theoretical result available showing that the
limit of the regularized solutions obtained by the approach of the
current paper by sending $\kappa\rightarrow 0$ are the \emph{finite
energy solutions} in the sense of \cite{GKLU}. A further study in
this aspect may provide more insights into the invisibility
cloaking.

In order to achieve the similar invisibility result for a general
$l^p$-ball shaped obstacles, we need to derive stability estimates
similar to (\ref{eq:stability 1}) for generally shaped obstacles
with small diameters. This is given by
Lemma~\ref{lem:justification1} below, proved using boundary integral
representation rather than separation of variables, and the
obstacles could be generally star-shaped. On the other hand, from a
practical viewpoint, we consider the scattering with multiple
scattering components and only some of the components are cloaked.
We shall show that the virtual reshaping takes effect only for those
cloaked components and the other uncloaked components remain
unaffected. Particularly, those perfectly cloaked components will be
invisible, even though there is scattering interaction between the
obstacle components. We are now in a position to present the key
lemma. In the sequel, we let $\mathbb{B}$ be a simply connected set
in $\mathbb{R}^3$ whose boundary is star-shaped with respect to the
origin of the form $\partial \mathbb{B}= \delta r_0(\theta)\theta$,
where $\theta\in\mathbb{S}^2, r_0(\theta)\in C^2(\mathbb{S}^2)$ and
$\delta>0$. Let $\mathbb{B}_0$ be the domain $\{x\in\mathbb{R}^3;
|x|<r_0(\theta)\}$.

\begin{Lemma}\label{lem:justification1}
Suppose $M_1\cap \mathbb{B}_0=\emptyset$, then we have
\begin{equation}\label{eq:est1}
A(\theta',\theta, k;M_1\cup \mathbb{B})=A(\theta',\theta,k;
M_1)+\mathcal{O}(\delta)\quad \mbox{as $\delta\rightarrow +0$}.
\end{equation}
\end{Lemma}

\begin{proof}
Let $\Phi(x,y)=e^{ik|x-y|}/(4\pi|x-y|)$ be the fundamental solution
to the Helmholtz operator $(\Delta+k^2)$. We know that
$u(x;k,\theta)\in
C^2(\mathbb{R}^3\backslash\overline{M_1\cup\mathbb{B}})\cap
C(\mathbb{R}^3\backslash (M_1\cup\mathbb{B}))$ and can be
represented in the form (see \cite{ColKre})
\begin{equation}\label{eq:representation}
\begin{split}
u(x; M_1\cup\mathbb{B})=e^{ik x\cdot \theta}+& \int_{\partial
M_1}\bigg\{\frac{\partial\Phi(x,y)}{\partial\nu(y)}-{i}\Phi(x,y)\bigg\}\varphi_1(y)\
ds(y) \\
+ & \int_{\partial\mathbb{B}}
\bigg\{\frac{\partial\Phi(x,y)}{\partial\nu(y)}-{i\eta}\Phi(x,y)\bigg\}\varphi_2(y)\
ds(y),
\end{split}
\end{equation}
where $\varphi_1\in C(\partial M_1)$ and $\varphi_2\in
C(\partial\mathbb{B})$ are density functions, and $\eta\neq 0$ is a
real coupling parameter. The densities $\varphi_1$ and $\varphi_2$
are unique solutions to the following integral equation (see
\cite{ColKre})
\begin{equation}\label{eq:n1}
\begin{split}
\varphi(x)+ & 2\int_{\partial
M_1}\bigg\{\frac{\partial\Phi(x,y)}{\partial\nu(y)}-{i}\Phi(x,y)\bigg\}\varphi_1(y)\
ds(y)\\
+ & 2\int_{\partial\mathbb{B}}
\bigg\{\frac{\partial\Phi(x,y)}{\partial\nu(y)}-{i\eta}\Phi(x,y)\bigg\}\varphi_2(y)\
ds(y)=-2 e^{i kx\cdot \theta},
\end{split}
\end{equation}
for $x\in\partial M_1\cup\partial\mathbb{B}$, where
$\varphi(x):=\varphi_1(x)$ for $x\in\partial M_1$ and
$\varphi(x):=\varphi_2(x)$ for $x\in\partial \mathbb{B}$.
%
%
We introduce the integral operators
\begin{align*}
& (S_1\varphi_1)(x)=2\int_{\partial M_1}\Phi(x,y)\varphi_1(y) d
s(y),\quad
(K_1\varphi_1)(x)=2\int_{\partial M_1}\frac{\partial\Phi(x,y)}{\partial\nu(y)}\varphi_1(y)ds(y)\\
&(S_2\varphi_2)(x)=2\int_{\partial\mathbb{B}}\Phi(x,y)\varphi_2(y) d
s(y),\quad
(K_2\varphi_2)(x)=2\int_{\partial\mathbb{B}}\frac{\partial\Phi(x,y)}{\partial\nu(y)}\varphi_2(y)ds(y),
\end{align*}
and set
\[
h_1(x):=-2e^{ik x\cdot\theta},\ \ x\in
\partial M_1;\qquad h_2(x):=-2e^{ik x\cdot \theta},\ \ x\in\partial\mathbb{B}.
\]
Then equation (\ref{eq:n1}) can be rewritten as
\begin{align}
& [\varphi_1+K_1\varphi_1-{i}S_1\varphi_1+K_2\varphi_2-{i}\eta
S_2\varphi_2](x)=h_1(x),\quad
x\in \partial M_1\label{eq:eq1}\\
&
[\varphi_2+K_2\varphi_2-{i\eta}S_2\varphi_2+K_1\varphi_1-{i}S_1\varphi_1](x)=h_2(x),\quad
x\in
\partial\mathbb{B}.\label{eq:eq2}
\end{align}
It is remarked that the integral operators involved in equations
(\ref{eq:eq1}) and (\ref{eq:eq2}) with weakly singular integral
kernels have to be understood in the sense of Cauchy principle
values and we refer to \cite{ColKre} and \cite{Mcl} for related
mapping properties. Clearly, $\varphi_1$ and $\varphi_2$ are
functions dependent on $\delta$. We next study their asymptotic
behaviors as $\delta\rightarrow +0$. To this end, we fix $\delta>0$
but being sufficiently small and take $\eta=\delta^{-1}$.

In the sequel, without loss of generality, we may assume that
$dist(\partial M_1,\partial\mathbb{B}_0)>c_0>0$, otherwise one can
shrink $\mathbb{B}_0$ to $1/2 \mathbb{B}_0$. By straightforward
calculations, it can be easily shown that
\begin{equation}\label{eq:relation1}
\|K_2-{i\eta}S_2\|_{C(\partial\mathbb{B})\rightarrow C(\partial
M_1)}=\mathcal{O}(\delta),\ \|K_1-{i}S_1\|_{C(\partial
M_1)\rightarrow C(\partial\mathbb{B})}=\mathcal{O}(1).
\end{equation}
Next, for $x\in\partial \mathbb{B}_0$, we define
\[
(K_0\phi)(x)=2\int_{\partial\mathbb{B}_0}\frac{\partial
\Phi_0(x,y)}{\partial\nu(y)}\phi(y)\ ds(y),\
(S_0\phi)(x)=2\int_{\partial\mathbb{B}_0}\Phi_0(x,y)\phi(y)\ ds(y),
\]
where $\phi\in C(\partial\mathbb{B}_0)$ and
$\Phi_0(x,y)=1/(4\pi|x-y|)$ is the fundamental solution to the
Laplace operator. It is known that both $S_0$ and $K_0$ are compact
operators in $C(\partial\mathbb{B}_0)$ (see \cite{ColKre2}). By
changing the integration to the boundary of the reference obstacle
$\partial \mathbb{B}_0$, we have
\begin{align*}
(S_2\varphi)(x)=&2\int_{\partial\mathbb{B}}
\frac{e^{ik|x-y|}}{|x-y|}\varphi(y) \
ds(y)=2\delta\int_{\partial\mathbb{B}_0}\frac{e^{ik\delta|x'-y'|}}{|x'-y'|}\varphi(\delta
y')\ ds(y'),\\
(K_2\varphi)(x)=& 2 \int_{\partial\mathbb{B}}
\partial\left(\frac{e^{ik|x-y|}}{|x-y|}\right)/\partial\nu(y)\varphi(y)\
ds(y)\\
=& 2 \int_{\partial\mathbb{B}_0}
\partial\left(\frac{e^{ik\delta |x'-y'|}}{|x'-y'|}\right)/\partial\nu(y')\varphi(\delta y')\
ds(y'),
\end{align*}
where $\varphi\in C(\partial\mathbb{B})$, $x,y\in\partial
\mathbb{B}$ and $x':=x/\delta, y'=y/\delta\in\partial \mathbb{B}_0$.
Then by using power series expansion of the exponential function
$e^{ik\delta|x-y|}$, we have by direct calculations
\begin{equation}\label{eq:relation2}
\|\frac{1}{\delta}S_2-S_0\|_{C(\delta\partial\mathbb{B}_0)\rightarrow
C(\partial\mathbb{B}_0)}=\mathcal{O}(\delta),\quad
\|K_2-K_0\|_{C(\delta\partial\mathbb{B}_0)\rightarrow
C(\partial\mathbb{B}_0)}=\mathcal{O}(\delta^2).
\end{equation}

By changing the integration to $\partial\mathbb{B}_0$ and using the
results in (\ref{eq:relation2}), we have from (\ref{eq:eq2}) that
\begin{equation}\label{eq:varphi2}
\varphi_2(\delta x)=[I+K_0-i
S_0+\mathcal{O}(\delta)]^{-1}[h_2(\delta
x)-(K_1-{i}S_1)\varphi_1(\delta {x})],\quad
{x}\in\partial\mathbb{B}_0.
\end{equation}
It is noted here that $(I+K_0-i S_0)$ is bounded invertible (see
\cite{ColKre2}). Then, plugging (\ref{eq:varphi2}) into
(\ref{eq:eq1}) and using the relations in (\ref{eq:relation1}), we
further have
\begin{equation}
\varphi_1=[I+K_1-\mathrm{i}S_1+\mathcal{O}(\delta)]^{-1}(h_1+\mathcal{O}(\delta)),
\end{equation}
which, by noting $I+K_1-\mathrm{i}S_1$ is invertible (see
\cite{ColKre}), gives
\begin{equation}\label{eq:error2}
\varphi_1=\tilde{\varphi}_1+\mathcal{O}(\delta),
\end{equation}
where
\[
\tilde{\varphi}_1=[I+K_1-\mathrm{i}S_1]^{-1}h_1.
\]
Furthermore, (\ref{eq:error2}) together with (\ref{eq:varphi2})
implies that
\begin{equation}\label{eq:error3}
\varphi_2=\mathcal{O}(1).
\end{equation}
Finally, by (\ref{eq:representation}) we know
\begin{equation}\label{eq:far-field}
\begin{split}
A(\theta',\theta,k; M_1\cup\mathbb{B})=&\frac{1}{4\pi}\int_{\partial
M_1}\bigg\{\frac{\partial e^{-iky\cdot
\theta'}}{\partial\nu(y)}-{i}e^{-iky\cdot\theta'}\bigg\}\varphi_1(y)\
ds(y)\\
+& \frac{1}{4\pi}\int_{\partial \mathbb{B}}\bigg\{\frac{\partial
e^{-iky\cdot
\theta'}}{\partial\nu(y)}-{i\eta}e^{-iky\cdot\theta'}\bigg\}\varphi_2(y)\
ds(y).
\end{split}
\end{equation}
Using the estimates in (\ref{eq:error2}) and (\ref{eq:error3}) to
(\ref{eq:far-field}) and changing the integration over
$\partial\mathbb{B}$ to $\partial\mathbb{B}_0$, we have
\[
A(\theta',\theta,k;M_1\cup\mathbb{B})=
A(\theta',\theta,k;M_1)+\mathcal{O}(\delta),
\]
where we have made use of the fact that
\[
A(\theta',\theta,k;M_1)=\frac{1}{4\pi}\int_{\partial
M_1}\tilde{\varphi}_1(y) \bigg\{\frac{\partial e^{-{i}k \theta'\cdot
y}}{\partial \nu(y)}-{i} e^{-{i}k \theta'\cdot y}\bigg\}\ ds(y).
\]
The proof is completed.
\end{proof}

\begin{Remark}\label{rem:nonsmooth}
If $\partial M_1\cup \partial\mathbb{B}$ is only Lipschitz
continuous (whence $r_0(\theta)\in C^{0,1}(\mathbb{S}^{2})$), one
can make use of the mapping properties of relevant boundary layer
potential operators presented in \cite{Mcl} and derive similar
estimate.
\end{Remark}

\begin{Proposition}\label{prop22}
Suppose that $M_1\cap B_{R_2,p}=\emptyset$. The cloaking device
$((M_1\cup B_{R_1,p})$ $\bigoplus((M_1\cup B_{R_1,p})^+, \hat g))$,
where $\hat{g}(x)$ is $g_1$ in (\ref{eq:1.6}) for $R_1<\|x\|_p<R_2$
and $g_0(x)$ for $x\in (M_1\cup B_{R_2,p})^+$, reshapes the obstacle
$M_1\cup B_{R_1,p}$ to $M_1\cup B_{R_0,p}$. Furthermore, the
limiting case with $R_0=0$ corresponds to the perfect cloaking of
$B_{R_1,p}$, namely
\begin{equation}
A(\theta',\theta,k;(M_1\cup B_{R_1,p})\bigoplus((M_1\cup
B_{R_1,p})^+, \hat g))=A(\theta',\theta,k; M_1).
\end{equation}
\end{Proposition}

\begin{proof}
Let $F$ be the transformation in (\ref{eq:1}) and let $\hat{F}$ be
the restriction of $F$ over $(M_1\cup B_{R_0})^+$. Clearly,
$\hat{g}=\hat{F}_* g_0$. By a similar argument as the proof of
Proposition~\ref{prop:1}, it is easily seen that $M_1\cup B_{R_1,p}$
is virtually reshaped to $M_1\cup B_{R_0,p}$ by the cloaking of
$\hat{g}$, i.e.,
\[
A(\theta',\theta,k; (M_1\cup B_{R_1,p})\bigoplus((M_1\cup
B_{R_1,p})^+, \hat g))=A(\theta',\theta,k; M_1\cup B_{R_0,p}).
\]

Next, by Lemma~\ref{lem:justification1},
\begin{equation}
A(\theta',\theta, k; M_1\cup B_{R_0,p})=A(\theta',\theta, k;
M_1)+\mathcal{O}(R_0)\quad \mbox{as $R_0\rightarrow+0$},
\end{equation}
and hence the limiting case with $R_0=0$ yields an ideal cloaking of
$B_{R_1,p}$.
\end{proof}

\begin{Remark}
Clearly, Proposition~\ref{prop22} implies a same invisibility result
for perfectly cloaking an $l^p$-ball as that in
Proposition~\ref{prop:invisible} for perfectly cloaking an Euclidean
ball.
\end{Remark}


\subsection{Virtual magnification by cloaking}

Let $0<R_0<R_1<R_2$, and let $B_{R_0,p}$ be the obstacle which we
intend to virtually magnify to $B_{R_1,p}$ by using a cloaking for
$B_{R_0,p}$ supported in $B_{R_2,p}\backslash B_{R_0,p}$ (see
Fig.~\ref{fig:3}). We define $\tau=R_1/R_0$ to be the
\emph{magnification ratio}.

\begin{figure}
\centering
  \includegraphics[width=9cm]{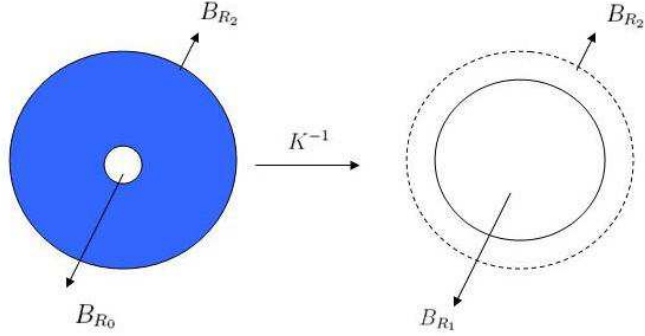}\\
  \caption{Illustration for magnification: the cloaking device $B_{R_0}\bigoplus (B_{R_0}, \hat{\hat{g}})$
  virtually reshapes $B_{R_0}$ to $B_{R_1}$.}\label{fig:3}
\end{figure}

Let $K: \mathbb{R}^3\backslash B_{R_1,p}\mapsto
\mathbb{R}^3\backslash B_{R_0,p}$ be defined by
\begin{equation*}
x:=K(y)=
\begin{cases}
\qquad y, & \mbox{for $\|y\|_p\geq R_2$},\\
\bigg(\frac{R_0-R_1}{R_2-R_1}R_2+\frac{R_2-R_0}{R_2-R_1}\|y\|_p\bigg)\frac{y}{\|y\|_p},\
& \mbox{for $R_0<\|y\|_p<R_2$}.
\end{cases}
\end{equation*}
It is verified directly that $K$ maps the $l^p$-annulus $R_1\leq
\|y\|_p\leq R_2$ to the $l^p$-annulus $R_0\leq \|y\|_p\leq R_2$.
Moreover, $K$ is bijective and both $K$ and $K^{-1}$ are (uniformly)
Lipschitz continuous. Set
\begin{equation}\label{eq:medium2}
\hat{\hat{g}}=K_* g_0,
\end{equation}
be the metric in $B_{R_0,p}^+$. Clearly, $\hat{\hat{g}}$ is
Euclidean outside $B_{R_2,p}$.
\begin{Proposition}\label{prop:21}
The cloaking device $B_{R_0,p}\bigoplus (B_{R_2,p}\backslash
B_{R_0,p}, \hat{\hat{g}})$ with $\hat{\hat{g}}$ defined in
(\ref{eq:medium2}), reshapes $B_{R_0,p}$ virtually to $B_{R_1,p}$.
That is, the physical obstacle $B_{R_0,p}$ with the cloaking
material $(B_{R_2,p}\backslash B_{R_0,p},\hat{\hat{g}})$ is
virtually magnified to the obstacle $B_{R_1,p}$ with a magnification
ratio $\tau:=R_1/R_0$.
\end{Proposition}

\begin{proof}
Let $u\in H_{loc}^1(B_{R_1,p}^+)$ be the solution to the Helmholtz
equation~(\ref{eq:Helmholtz})-(\ref{eq:Dirichlet}) associated with
the obstacle $B_{R_1,p}$. Define $\tilde{u}=K_* u\in
H_{loc}^1(B_{R_0,p}^+)$. Again, by the invariance of the Helmholtz
equation under transformation together with the fact that $\tilde
u|_{\partial B_{R_0,p}}=u|_{\partial B_{R_1,p}}=0$, we see
$\tilde{u}$ is the scattering solution corresponds to
$B_{R_0,p}\bigoplus (B_{R_0,p}^+, \hat{\hat{g}})$. That is,
\[
A(\theta',\theta,k; B_{R_0,p}\bigoplus (B_{R_0,p}^+,
\hat{\hat{g}}))=A(\theta',\theta,k; B_{R_1,p}).
\]
\end{proof}

In Proposition~\ref{prop:1}, we use $F$ to compress the vacuum to
achieve a transformation-based minification device, whereas in
Proposition~\ref{prop:21}, we use $K$ to loosen up the vacuum to
achieve a transformation-based magnification device. Note that
$R_2>R_1$, the cloaking device is of size larger than the virtual
obstacle image, though the virtual obstacle could be of size
arbitrarily close to the cloaking device. Hence, the cloaking in
Proposition~\ref{prop:21} is not of magnification in the real sense.
However, our magnification result is still of particular practical
interests, e.g., if one is only interested in recovering an obstacle
without knowing a priori that it is cloaked, then the scattering
reconstruction will give a virtually magnified obstacle. On the
other hand, we would like to mention that in \cite{MNMCJ,MNMP,NMM},
it is demonstrated that a coated cylindrical core can be extended
beyond the cloaking shell into the matrix, where the cloaking
material must be negative refractive indexed, namely, the
corresponding metric $g$ has negative eigenvalues. A general
strategy is presented in \cite{LeoPhi} on how to devise a negative
refractive indexed (NRI) cloaking by using the transformation
optics. There, the transformation $F$ is neither injective nor
orientation-preserving, which maps a right-handed medium to
left-handed medium. Based on NRI cloaking, it is shown in
\cite{LYGM,YCLM} that one can virtually reshape a cylindrical
perfect conductor of size bigger than the cloaking device. However,
all the aforementioned results are essentially based on exerting
transformation directly to the analytical solutions, which is not of
the main theme of the present paper.

\subsection{Virtually reshaping acoustic obstacles by cloaking}\label{sect:general acoustic
reshaping}

Our discussion so far has been mainly concerned with the
minification and magnification of obstacles by cloaking. Clearly, we
may consider virtually reshaping an obstacle arbitrarily provided a
suitable transform can be found with which we can make essential use
of the transformation invariance of the Helmholtz equation.  Let $M$
be an obstacle with $m$ pairwise disjoint simply connected
components $M_l, l=1,2,\ldots, m$, i.e., $M=\bigcup_{l=1}^{m} M_l$.
Let $M_l\subset M_{l}'$ and $M_l'\cap M_{l'}'=\emptyset$, for
$l,l'=1,2,\ldots,m$ and $l\neq l'$. Set $M'=\bigcup_{l=1}^{m} M_l'$.
Let $\widetilde{M}=\bigcup_{l=1}^m \widetilde{M}_l$ be another
obstacle with $\widetilde{M}_l\subset {M}_{l'}$. Suppose there exist
\begin{equation*}
F_l: \mathbb{R}^3\backslash \widetilde{M}_l\mapsto
\mathbb{R}^3\backslash M_l,\quad l=1,2,\ldots, m
\end{equation*}
such that $F_l$ is orientation-preserving and invertible with $F_l$
and $F^{-1}_l$ Lipschitz continuous, and $F_l=id$ outside $M_l'$.
Set $g_l=(F_l)_* g_0$ and let $M_l\bigoplus(M_l'\backslash M_l,
g_l)$ be a cloaking device for $M_l$.

\begin{Theorem}\label{thm:acoustic reshaping}
The cloaking device $M'$ virtually reshapes the obstacle $M$ to
$\widetilde{M}$. That is,
\[
A(\theta',\theta,k; M\bigoplus(M'\backslash M,
g'))=A(\theta',\theta, k; \widetilde{M}),
\]
where $g'$ is $g_l$ in $M_l'\backslash M_l$.
\end{Theorem}
The proof is already clear from our earlier discussion on
minification and magnification. We have several important
consequences of the theorem.

\begin{Remark}
Suppose that some of the components of $M$ are uncloaked, say $M_l$
for $1\leq l\leq m'<m$, and this corresponds to taking
$M_l=M_l'=\widetilde{M}_l$ and $F_l=id$ for $l=1,2,\ldots,m'$.
\end{Remark}

\begin{Remark}\label{rem:invisibility}
If for some $M_l$ being star-shaped w.r.t. certain point, and the
transformation $F_l^{-1}$ shrinks $M_l$ only in the radial direction
to $\widetilde{M}_l$, then the case with $\widetilde{M}_{l}$
degenerated to a single point corresponds to an ideal cloaking for
$M_{l}$. By using a similar argument as that for
Proposition~\ref{prop22} together with the estimate in
Lemma~\ref{lem:justification1}, one has that $M_{l}$ is invisible to
detection. In fact, by repeating the argument, the same conclusion
holds when there are more than one obstacle component is perfectly
cloaked.
\end{Remark}

It is noted that in Remark~\ref{rem:invisibility}, the perfectly
cloaked obstacle components are required to be star-shaped, and this
is because we need to make use of the estimate in
Lemma~\ref{lem:justification1} to achieve the invisibility. In order
to show the prefect cloakings of more generally shaped obstacles,
one may need different thoughts.

In the rest of this section, we shall indicate that all our previous
results on virtual reshaping in space dimension three can be
straightforwardly extended to the two dimensional case. In fact, for
two dimensional scattering problem, the (positive definite) acoustic
density $\sigma\in L^\infty(M^+)^{2\times 2}$ also transforms
according to (\ref{eq:trans sigma}). Therefore, the reshaping result
presented in Theorem~\ref{thm:acoustic reshaping} is still valid in
$\mathbb{R}^2$. In order to achieve invisibility for perfect
cloaking of star-shaped obstacles in $\mathbb{R}^2$, one needs to
show a similar estimate to Lemma~\ref{lem:justification1}. Indeed,
replacing $\Phi(x,y)$ by the first kind Hankel function
$\frac{i}{4}H_0^{(1)}(k|x-y|)$ of order zero in the proof of
Lemma~\ref{lem:justification1} and using the corresponding mapping
properties of the integral operators involved (see \cite{ColKre}),
one can obtain by similar arguments the following estimate to the
scattering problem in $\mathbb{R}^2$ (see (\ref{eq:est1}) for
comparison),
\begin{equation}\label{eq:estimate2}
A(\theta',\theta,k;M_1\cup \mathbb{B})=A(\theta',\theta,k;
M_1)+\mathcal{O}(|\log\delta|^{-1})\quad \mbox{as $\delta\rightarrow
+0$}.
\end{equation}
Obviously, with (\ref{eq:estimate2}) one can show that the perfect
cloaking of a star-shaped obstacles in $\mathbb{R}^2$ makes it
invisible to detection.

\section{Virtual reshaping for electromagnetic scattering}\label{sect:em
reshaping}

\subsection{The Maxwell's equations}
We define Maxwell's equations for the scatterer $M\bigoplus(M^+,g)$
as the one introduced in Section~\ref{subset:11}. Using the metric
$g$, we define a (positive definite) electric permittivity tensor
$\varepsilon$ and magnetic permeability tensor $\mu$ by
\begin{equation}\label{eq:f}
\varepsilon^{ij}=\mu^{ij}=|g|^{1/2}g^{ij}\quad \mbox{on\ $M^+$}.
\end{equation}
It is clear that $\varepsilon=(\varepsilon_{ij})_{i,j=1}^3$ and
$\mu=(\mu_{ij})_{i,j=1}^3$ are invariantly defined and transform as
a product of a $(+)$-density and a contravariant symmetric
two-tensor with the same rule as that for acoustic density $\sigma$
in (\ref{eq:trans sigma}). We consider the scattering due to the
scatterer $M\bigoplus(M^+,g)$ corresponding to some incident wave
field. The resulting total electric and magnetic fields, $E$ and $H$
in $M^+$, are defined as differential 1-forms, given in some local
coordinates by
\[
E=E_j\ d x^j,\qquad H=H_j\ d x^j.
\]
Here and in the following, we use Einstein's summation convention,
summing over indices appearing both as sub- and super-indices in
formulae. Then $(E,H)$ satisfies Maxwell's equations on $(M^+,g)$ at
frequency $k$
\begin{equation}\label{eq:Maxwell equations}
d E=i k *_g H,\quad d H=-ik *_g E,
\end{equation}
where $*_g$ denote the Hodge-operator on $1$-forms given by
\[
*_g(E_j\ dx^j)=\frac 1 2 |g|^{1/2}g^{jl}E_j s_{lpq} d x^p\wedge d
x^q=\frac 1 2 \varepsilon^{jl} E_j s_{lpq} d x^p\wedge d x^q,
\]
with $s_{lpq}$ denoting the Levi-Civita permutation symbol, and
$s_{lpq}=1$ (resp. $s_{lpq}=-1$) if $(l,p,q)$ is an even (resp. odd)
permutation of $(1,2,3)$ and zero otherwise. By introducing, for
$H=H_j dx^j$, the notation
\[
(\mbox{curl}\ H)^l=s^{lpq} \frac{\partial}{\partial x^p}H_q,
\]
the exterior derivative may then be written as
\[
d(H_q dx^q)=\frac{\partial H_q}{\partial x^p} d x^p\wedge d
x^q=\frac 1 2 (\mbox{curl}\ H)^l s_{lpq} dx^p\wedge dx^q.
\]
Hence, in a fix coordinate, the Maxwell's equations (\ref{eq:Maxwell
equations}) can be written as
\begin{equation}\label{eq:maxwell 2}
(\mbox{curl}\ E)^l=i k \mu^{jl} H_j,\quad (\mbox{curl}\
H)^l=-ik\varepsilon^{jl} E_j.
\end{equation}
Without loss of generality, we take the incident fields to be the
normalized time-harmonic electromagnetic plane waves,
\[
{E}^i(x):=\frac{{i}}{k}\mbox{curl}\ \mbox{curl}\,p\, e^{{i}k
{x}\cdot \theta},\quad {H}^i(x):= \mbox{curl}\,p\,e^{\mathrm{i}k
{x}\cdot \theta},
\]
where $p\in \mathbb{R}^3$ is a polarization. As usual, the radiation
fields are assumed to satisfy the \emph{Silver-M\"uller radiation
condition}. To complete the description, we further assume that the
obstacle $M$ is perfectly conducting, and we have the following two
types of boundary conditions on $\partial M$: the perfect electric
conductor (PEC) boundary condition
\begin{equation*}
\nu\times E|_{\partial M}=0,
\end{equation*}
or the perfect magnetic conductor (PMC) boundary condition
\begin{equation*}
\nu\times H|_{\partial M}=0,
\end{equation*}
where $\nu$ is the Euclidean normal vector of $\partial M$.

We shall work with $\varepsilon, \mu\in (L^\infty(M^+))^{3\times
3}$. It is convenient to introduce the following Sobolev spaces
\begin{align*}
 H(\mbox{curl};\Omega)&=\{u\in L^2(\Omega)^3;\ \mbox{curl}\, u\in
L^2(\Omega)\},\\
 H_{loc}(\mbox{curl}; M^+)&=\{u\in\mathscr{D}'(M^+)^3;\  u\in
H(\mbox{curl}; M^+\cap B_\rho)\\
&\hspace*{1.5cm}\mbox{for each finite $\rho$ with $M\subset
B_\rho$}\}.
\end{align*}
Then it is known that there exists a unique solution $(E,H)\in
H_{loc}(\mbox{curl}; M^+)\oplus H_{loc}(\mbox{curl}; M^+)$ to the
electromagnetic scattering problem. Moreover, the solution
$E(x,k,p,\theta)$ admits asymptotically as $|x|\rightarrow +\infty$
the development (see \cite{ColKre})
\begin{equation}
E(x,k,p,\theta)=E^i(x)+\frac{e^{ik
|x|}}{|x|}E_\infty(\theta',k,p,\theta)+\mathcal{O}(\frac{1}{|x|^2}),
\end{equation}
where $\theta'=x/|x|\in\mathbb{S}^{2}$. The analytic function
$E_\infty(\theta',k,p,\theta)$ is known as the \emph{electric
far-field pattern}. Similar to Definition~\ref{def:acoustic}, we
introduce

\begin{Definition}\label{def:em}
We say that $M\bigoplus (M^+,g)$ (virtually) reshapes the obstacle
$M$ to another obstacle $\widetilde{M}$, if the electric far-field
patterns coincide for $M\bigoplus (M^+,g)$ and $\widetilde{M}$, i.e.
\[
E_\infty(\theta',k,p,\theta;M\bigoplus
(M^+,g))=E_\infty(\theta',k,p,\theta; \widetilde{M}).
\]
\end{Definition}
We would also like to remark that the cloaking device in
Definition~\ref{def:em} can also be written as $M\bigoplus
(M^+,\varepsilon,\mu)$ according to the correspondence (\ref{eq:f}),
where $\varepsilon$ and $\mu$ are respectively, electric
permittivity and magnetic permeability for the cloaking medium.

\subsection{Virtually reshaping electromagnetic obstacles by cloaking}

We consider the virtual reshaping for electromagnetic obstacles by
cloaking. Let $M=\bigcup_{l=1}^{m} M_l$, $M'=\bigcup_{l=1}^{m}
M_l'$, $\widetilde{M}=\bigcup_{l=1}^m \widetilde{M}_l$ and $F_l,
l=1,2,\ldots,m$ be those introduced in Section~\ref{sect:general
acoustic reshaping}. Furthermore, we assume that $F_l$ is
normal-preserving in the sense that
\[
\tilde{\nu}_l=\nu_l\circ F_l,\quad l=1,2,\ldots,m,
\]
where $\tilde{\nu}_l$ and $\nu_l$ are, respectively, the Euclidean
normals to $\partial\widetilde{M}_l$ and $\partial M_l$. E.g., if
$\widetilde{M}_l$ and $M_l$ are both star-shaped w.r.t. the origin,
say $\partial\widetilde{M}_l=\tilde{r}(\theta)\theta$ and $\partial
M_l=r(\theta)\theta$ with $\tilde{r}/r=c$ being some constant, then
$F_l$ is normal-preserving since one has
\[
\tilde{\nu}_l|_{\tilde{r}(\theta)\theta}=\nu_l|_{r(\theta)\theta}=\frac{r(\theta)\theta-\mbox{Grad}
r}{\sqrt{r^2+|\mbox{Grad} r|^2}}.
\]
Particularly, if $\widetilde{M}_l$ is $l^p$-ball shaped, the
transformation of the following form
\[
F(y)=(a+b\|y\|_p)\frac{y}{\|y\|_p},
\]
is normal-preserving, which transforms an $l^p$-ball of radius
$\tilde r$ into another $l^p$-ball of radius $r=a+b\tilde r$.

Concerning the virtual reshaping, we have

\begin{Theorem}\label{thm:magnetic reshaping}
The cloaking device $M'$ virtually reshapes the obstacle $M$ to
$\widetilde{M}$. That is,
\[
E_\infty(\theta',k,p,\theta; M\bigoplus(M'\backslash M,
g'))=E_\infty(\theta',k,p,\theta; \widetilde{M}),
\]
where $g'$ is $g_l$ in $M_l'\backslash M_l$.
\end{Theorem}
\begin{proof}
Let $F: \mathbb{R}^3\backslash
\widetilde{M}\mapsto\mathbb{R}^3\backslash M$ be such that
$F|_{M_l'\backslash \widetilde{M}_l}=F_l|_{M_l'\backslash
\widetilde{M}_l}, l=1,2,\ldots,m$ and $F=id$ over
$\mathbb{R}^3\backslash M'$. Let $(E,H)\in H_{loc}(\mbox{curl};
\widetilde{M}^+)\oplus H_{loc}(\mbox{curl};$ $\widetilde{M}^+)$ be
the unique scattering solution corresponding to the perfect
conducting obstacle $\widetilde{M}$. Define $\hat{E}=F_* E$ and
$\hat H=F_* H$. Clearly, $(\hat{E},\hat H)\in H_{loc}(\mbox{curl};
M^+)\oplus H_{loc}(\mbox{curl}; M^+)$ according to our requirements
on the mappings $F_l$'s, $l=1,2,\ldots,m$. Moreover, noting $F_l$'s,
$1\leq l\leq m$, are normal-preserving, we know $\nu\times
\hat{E}|_{\partial M}=0$ (resp. $\nu\times \hat{H}|_{\partial M}=0$)
if $\widetilde{M}$ is a perfectly electric conducting obstacle
(resp. perfectly magnetic conducting obstacle). Hence, $(\hat E,
\hat H)$ is the unique solution corresponding to the cloaking device
$M\bigoplus(M'\backslash M, g')$. Therefore, we have
\[
E_\infty(\theta',k,p,\theta; M\bigoplus(M'\backslash M,
g'))=E_\infty(\theta',k,p,\theta; \widetilde{M}).
\]
\end{proof}

With Theorem~\ref{thm:magnetic reshaping}, all the virtual
minification and magnification results for acoustic obstacle
scattering can be straightforwardly extended to the electromagnetic
obstacle scattering. In order to obtain similar invisibility results
for a perfectly conducting obstacle when some of its star-shaped
components are perfectly cloaked, we need a lemma similar to
Lemma~\ref{lem:justification1} in the following for electromagnetic
scattering.

\begin{Lemma}\label{lem:justification2}
Let $M_1$, $\mathbb{B}_0$ and $\mathbb{B}$ be the same as those in
Lemma~\ref{lem:justification1}, then we have
\begin{equation}\label{eq:est2}
E_\infty(\theta',k,p,\theta; M_1\cup
\mathbb{B})=E_\infty(\theta',k,p,\theta;
M_1)+\mathcal{O}(\delta)\quad \mbox{as $\delta\rightarrow +0$}.
\end{equation}
\end{Lemma}

\begin{proof}
We first introduce the space $T^{0,\alpha}(\partial M_1),
0<\alpha\leq 1$, consisting of the uniformly H\"older continuous
tangential fields $a$ equipped with the H\"older norm, and
$T^{0,\alpha}_d(\partial M_1)=\{a\in T^{0,\alpha}(\partial M_1);
\mbox{Div}\, a\in C^{0,\alpha}(\partial M_1)\}$. Similarly, one can
introduce $T^{0,\alpha}_d(\partial M_1\cup \partial\mathbb{B})$. We
know the solution $(E,H)\in
C^{0,\alpha}(\mathbb{R}^3\backslash(M_1\cup \mathbb{B}))\oplus
C^{0,\alpha}(\mathbb{R}^3\backslash(M_1\cup \mathbb{B}))$ and can be
expressed (cf. \cite{ColKre})
\begin{align*}
E(x)=E^i(x)&+\mbox{curl}\int_{\partial M_1} a_1(y)\Phi(x,y)
ds(y)+i\int_{\partial M_1}\nu(y)\times (S_0^2 a)(y)\Phi(x,y)ds(y)\\
&+\mbox{curl}\int_{\partial \mathbb{B}} a_2(y)\Phi(x,y)
ds(y)+i\eta\int_{\partial \mathbb{B}}\nu(y)\times (S_0^2
a_2)(y)\Phi(x,y)ds(y)
\end{align*}
and $H(x)=\mbox{curl} E(x)/ik$, where $\eta\neq 0$ is a real
coupling parameter. Here, $S_0$ is the operator as defined in the
proof of Lemma~\ref{lem:justification1} but with the integration
domains changed according to the context and the densities $a_1\in
T_d^{0,\alpha}(\partial M_1), a_2\in
T_d^{0,\alpha}(\partial\mathbb{B})$ satisfy
\begin{equation}\label{eq:l1}
(a+\mathcal{M}_1a_1+i\mathcal{N}_1PS_0^2 a_1+\mathcal{M}_2 a_2 +
i\eta\mathcal{N}_2 P S_0^2 a_2)(x)=V(x),\ \ x\in\partial
M_1\cup\partial\mathbb{B}
\end{equation}
where $a(x):=a_1(x)$ for $x\in\partial M_1$ and $a(x):=a_2(x)$ for
$x\in\partial \mathbb{B}$, and $V(x):=-2 \nu\times E^i(x)$ for PEC
obstacle and $V(x):=-2 \nu\times H^i(x)$ for PMC obstacle. The
operators involved in (\ref{eq:l1}) are respectively given by
\begin{align*}
(\mathcal{M}_1a_1)(x):=& 2\int_{\partial M_1}\nu(x)\times
\mbox{curl}_x\{a_1(y)\Phi(x,y)\}\ ds(y),\\
(\mathcal{N}_1b_1)(x):=& 2 \nu(x)\times\mbox{curl}\
\mbox{curl}\int_{\partial M_1}\nu(y)\times
b_1(y)\Phi(x,y)\ ds(y),\\
(\mathcal{M}_2a_2)(x):=& 2\int_{\partial \mathbb{B}}\nu(x)\times
\mbox{curl}_x\{a_2(y)\Phi(x,y)\}\ ds(y),\\
(\mathcal{N}_2b_2)(x):=& 2 \nu(x)\times\mbox{curl}\
\mbox{curl}\int_{\partial \mathbb{B}}\nu(y)\times
b_2(y)\Phi(x,y)\ ds(y),\\
P c:=& (\nu\times c)\times c.
\end{align*}
We again refer to \cite{ColKre2,ColKre} for relevant mapping
properties of the above operators. Finally, a similar asymptotic
analysis to that implemented in the proof of
Lemma~\ref{lem:justification1}, one can complete the proof.
\end{proof}

Clearly, with Lemma~\ref{lem:justification2}, we have similar
invisibility result for electromagnetic scattering as those remarked
in Remark~\ref{rem:invisibility} for acoustic scattering.

\section*{Acknowledgement}

The author would like to thank Professor Gunther Uhlmann of the
Department of Mathematics, University of Washington for a lot of
stimulating discussions. He would also like to thank the
constructive comments from two anonymous referees, which have led to
significant improvements of the results of this paper. The work is
partly supported by NSF grant, FRG DMS 0554571.

\end{document}